\documentclass{llncs}
\usepackage{epsfig}
\usepackage{theorem}
\usepackage{amsmath}
\usepackage{amssymb}
\usepackage{appendix}
\usepackage{psfrag}
\usepackage{graphicx}
\usepackage{slashbox}
\usepackage{array}
\usepackage{etoolbox}
\apptocmd{\thebibliography}{\scriptsize}{}{}

\usepackage{algorithm}
\usepackage{algorithmic}
\usepackage{xcolor}
\usepackage{epstopdf}

\usepackage{latexsym}
\usepackage{graphics}
\usepackage{epic}
\usepackage{eepic}
\usepackage{color}
\newcommand{\even}{\mathcal{E}}
\newcommand{\un}{\mathcal{U}}
\newcommand{\EE}{\mathcal{E}}
\newcommand{\odd}{\mathcal{O}}
\newcommand{\U}{\mathcal{U}}
\newcommand{\A}{\mathcal{A}}

\newcommand{\p}{\mathcal{P}} 
\newcommand{\X}{\mathcal{X}}
\newcommand{\Ecomp}{{sink component} }
\newcommand{\Ucomp}{{\em non-sink component} }
\newcommand{\rmm} {{\em rmm }}
\newcommand{\REM}[1]{}
\newcommand{\shP}{\mbox{{\sf \#P}}}
\newcommand{\shBPM}{\mbox{{\sf \#BPM}}}
\newcommand{\shRMM}{\mbox{{\sf \#RMM}}}
\newcommand{\fpras}{\mbox{{\sf FPRAS}}}

\newenvironment{appendix-theorem}[1]{\vspace{\theorempreskipamount}\noindent{\bf Theorem~#1~} \em }{\vspace{\theorempostskipamount}}
\newenvironment{appendix-lemma}[1]{\vspace{\theorempreskipamount}\noindent{\bf Lemma~#1~} \em }{\vspace{\theorempostskipamount}}
\newenvironment{appendix-corollary}[1]{\vspace{\theorempreskipamount}\noindent{\bf Corollary~#1~} \em }{\vspace{\theorempostskipamount}}
\title{Rank-Maximal Matchings -- Structure and Algorithms}
\author{Pratik Ghosal$^1$
 \and
  Meghana Nasre$^2$
\and
Prajakta Nimbhorkar\thanks{Part of the work has been done while the
author was on a sabbatical to the Institute of Mathematics of the Czech Academy
of Sciences, Prague.}$^3$
 }

\institute{University of Wroc{\l}aw, Poland
 {\tt (pratikghosal20082@gmail.com)}
\and
Indian Institute of Technology Madras 
{\tt (meghana@cse.iitm.ac.in)}
\and
Chennai Mathematical Institute 
 {\tt (prajakta@cmi.ac.in)}
}

\begin{document}
\maketitle
\begin{abstract}
Let $G = (\A \cup \p, E)$ be a bipartite graph  where $\A$ denotes
a set of agents, $\p$ denotes a set of posts and  ranks 
on the edges
denote preferences of the agents over posts. A matching 
$M$ in $G$ is rank-maximal if it matches the maximum number of applicants to 
their top-rank post, subject to this, 
the maximum number of applicants to their second rank post and so on.

In this paper, we develop a {\em switching graph} characterization of 
rank-maximal matchings, which is a useful tool that encodes all rank-maximal 
matchings in an instance.
The characterization leads to simple and efficient algorithms for several 
interesting problems. In particular, we give an efficient algorithm to compute 
the set of {\em rank-maximal pairs} in an
instance. We show that the problem of counting the number of rank-maximal 
matchings is $\#P$-Complete and also give an FPRAS for the problem. Finally, we
consider the problem of 
deciding whether a rank-maximal matching is {\em popular} among all the 
rank-maximal matchings in a given instance, and give an efficient algorithm for
the problem.
 
\end{abstract}

\section{Introduction}
We consider the problem of matching applicants to posts where applicants have 
preferences over posts. This problem is motivated by several important 
real-world applications like allocation of graduates to training positions 
\cite{HZ79} and families to government housing \cite{Yuan96}.
The input to the problem is a bipartite graph $G = (\A \cup \p, E)$, 
where $\A$ is a set of applicants, $\p$ is a set of posts,
and the set $E$ can be partitioned as $E = E_1 \cup \ldots \cup E_r$, where 
$E_i$ contains the edges of rank~$i$. An edge $(a,p) \in E_i$ if $p$ is an 
$i$th choice of $a$.
An applicant $a$ prefers a post $p$ to $p'$ if, for some $i < j$, 
$(a, p) \in E_i$ and $(a, p') \in E_j$. Applicant $a$ is indifferent between 
$p$ and $p'$ if $i = j$.
This ranking of posts by an applicant is called {\em the preference list} of the applicant.
When applicants can be indifferent between posts, preference lists are said to 
contain ties, else preference lists are strict.

The problem of matching under one-sided preferences 
has received lot of attention and there exist
several notions of optimality like pareto-optimality~\cite{ACMM04}, 
rank-maximality~\cite{IKMMP06}, popularity~\cite{AIKM07}, and fairness.
We
focus on the well-studied notion of {\em rank-maximal} matchings which are
guaranteed to exist in any instance. 
 Rank-maximality was first studied under the name of {\em greedy matchings} by 
Irving~\cite{Irving03},
who also gave an algorithm for
computing such matchings in case of strict lists. A rank-maximal matching 
matches maximum number of applicants to their rank~$1$ posts, subject to 
that, maximum number of applicants to their rank~$2$ posts and so on. 
Irving~et~al.\cite{IKMMP06} gave an $O(\min(n+r, r\sqrt{n})m)$ time algorithm 
to compute a rank-maximal matching. This algorithm\cite{IKMMP06} not only works for strict case, but also for tied case.
Here $n = |\A| + |\p|$, $m = |E|$, and $r$ denotes the maximal rank in the 
instance. The weighted and capacitated versions of this problem have been
studied in \cite{KS06} and \cite{Paluch13} respectively.

In this paper, we study the structure of the rank-maximal matchings using the 
notion of a {\em switching graph}. 
This notion was introduced in the context of {\em popularity} which is an alternative 
criterion of optimality in the one-sided preferences model. See 
\cite{AIKM07} for a definition of popular matchings.
McDermid and Irving~\cite{MI11} studied the switching graph of popular matchings
for strict instances, and 
Nasre~\cite{Nasre13} extended it to the case of ties. This characterization
has turned out to be useful for several problems like counting the number of
popular matchings in strict instances, computing an {\em optimal} popular 
matching, developing an optimal manipulation strategy for an agent etc.

It is natural to extend the switching graph characterization
to analyze rank-maximal matchings.
Besides being interesting in its own right,
it turns out to be useful in answering several natural 
questions. For instance, given instance $G = (\A \cup \p, E)$,
is there a rank-maximal matching in $G$ which matches 
an applicant $a$ to a particular post $p$? 
Is a rank-maximal matching preferred 
by a majority of applicants over other rank-maximal matchings in the instance?
We show the following new results in this paper:
\vspace{-0.05in}
\begin{itemize}
\item A switching graph characterization of the rank-maximal matchings problem, 
and its properties, using which, we answer the questions mentioned above.
\item An efficient algorithm for computing 
the set of {\em rank-maximal pairs}. An
edge $(a, p) \in E$ is a rank-maximal pair if there exists a rank-maximal
matching in $G$ that matches $a$ to $p$. 
\item 
We show that 
the problem of counting the number of rank-maximal matchings is \shP-complete
even for strict preference lists. We then give an \fpras\ for the problem by  
reducing it to the problem of counting the number of perfect 
matchings in a bipartite graph. 
\item 
In order to choose one among possibly several rank-maximal matchings in a given
instance $G$,
we consider the question of
finding a rank-maximal matching that is popular among all the rank-maximal 
matchings in $G$. We call such a matching a {\em popular rank-maximal matching}.
We show that, given
a rank-maximal matching, it can be efficiently checked whether it is a popular
rank-maximal matching. If not, we output a rank-maximal matching
which is more popular than the given one. 
\end{itemize}
We remark that the switching graph is a weighted directed graph constructed with respect to
a particular matching. In case of popular matchings, 
it is known from \cite{AIKM07} that, there are at most two distinct ranked 
posts in an applicant's preference list,
to which he can be matched in any popular matching. This results in a switching
graph with edge-weights $\{+1,-1,0\}$. In case of rank-maximal matchings, the
situation becomes more interesting since an applicant can be matched to one among
several distinct ranked posts, and the edge-weights in the switching graph could be 
arbitrary. Surprisingly, the characterization still turns out to be similar to 
that of popular matchings, although the proofs are significantly different. We
expect  that the switching graph  will find several applications apart from those shown in this 
paper.

\vspace{-0.1in}
\section{Preliminaries}\label{sec:prelim}
\vspace{-0.05in}
A matching $M$ of $G$ is a subset of edges, no two of which share an end-point.
For a matched vertex $u$, we denote by $M(u)$ its partner in $M$.

\paragraph{Properties of maximum matchings in bipartite graphs:} 
Let $G = (\A \cup \p, E)$ be a bipartite graph and let $M$ be a maximum matching in $G$.
The matching $M$ defines a partition of the vertex set $\A \cup \p$ into three 
disjoint sets, defined below:
\begin{definition}[Even, odd, unreachable vertices]\label{def:eou}
A vertex $v \in \A \cup \p$ is \emph {even} (resp. \emph {odd}) if there is an 
even (resp. odd) length alternating path with respect to $M$ from an unmatched 
vertex to $v$.
A vertex $v$ is \emph {unreachable} if there is no alternating path from an unmatched vertex to $v$.
\end{definition}
The following lemma is well-known in matching theory; see \cite{GGL95new} or \cite{IKMMP06} for a proof.

\begin{lemma}[\cite{GGL95new}]
\label{lem:node-class}
Let $\EE$, $\odd$, and $\U$ be the sets of even, odd, and unreachable vertices 
defined by a maximum matching $M$ in $G$. Then,
\begin {itemize}
\item [(a)] $\EE$, $\odd$, and $\U$ are disjoint, and are the same for all
 the maximum matchings in $G$.
\item [(b)] In any maximum matching of $G$, every vertex in $\odd$ is matched with a vertex in
$\EE$, and every vertex in $\U$ is matched with another vertex in $\U$.
The size of a maximum matching is $|\odd| + |\U|/2$.
\item  [(c)] No maximum matching of $G$ contains an edge with one end-point
 in $\odd$ and the other in $\odd \cup \U$.
Also, $G$ contains no edge with one end-point in $\EE$ and the other in $\EE \cup \U$.
\end {itemize}
\end{lemma}

\paragraph{Rank-maximal matchings:}
An instance of the rank-maximal matchings problem consists of a bipartite
graph $G=(\A\cup\p,E)$, where $\A$ is a set of applicants, $\p$ is
a set of posts, and 
$E$ can be
partitioned as $E_1 \cup E_2 \cup \ldots \cup E_r$. Here $E_i$ denotes
the edges of rank~$i$, and 
$r$ denotes the maximum rank any applicant assigns to a post. An edge $(a,p)$
has rank $i$ if $p$ is an $i$th choice of $a$. 

\begin{definition}[Signature]\label{def:sig}
The {\em
signature} of a matching $M$ is defined as an $r$-tuple $ \rho (M) = (x_1,\ldots,x_r)$
where, for each $1\leq i \leq r$, $x_i$ is the number of applicants who are 
matched to their $i$th rank post in $M$. 
\end{definition}
Let $M$, $M'$ be two matchings in $G$, with signatures 
$\rho(M) = (x_1, \ldots, x_r)$
and $\rho(M') = (y_1, \ldots, y_r)$.
Define $M \succ M'$ if $x_i = y_i$ for $1 \le i < k \leq r$ and $x_k >  y_k$.

\begin{definition}[Rank-maximal matching]\label{def:rmm}
A matching $M$ in $G$ is  rank-maximal if 
$M$ has the maximum signature under the above ordering $\succ$.
\end{definition}
Observe that all the rank-maximal matchings in an instance 
have the same cardinality and the same signature.

\paragraph{Computing Rank-maximal Matchings:}
Now we recall Irving~et~al.'s algorithm~\cite{IKMMP06} for computing a 
rank-maximal matching in a given instance $G=(\A\cup\p,E_1\cup\ldots\cup E_r)$.
Recall that $E_i$ is the set of edges of rank $i$.
For the sake of convenience, for each applicant $a$, we add a dummy last-resort
post $\ell(a)$ at rank $r+1$ in $a$'s preference list, and refer to the modified
instance as $G$. This ensures that every rank-maximal matching is $\A$-complete
i.e. matches all the applicants. 

Let $G_i=(\A\cup\p, E_1\cup\ldots\cup E_i)$.
The algorithm starts with $G_1'=G_1$ and any maximum matching $M_1$ in $G'_1$.

\noindent\fbox{
\parbox{\textwidth}{For $i=1$ to $r$ do the following and output $M_{r+1}$:
\begin{enumerate}
\item Partition the vertices in $\A\cup\p$ into even, odd, and unreachable
as in Definition \ref{def:eou} and call these sets $\EE_i,\odd_i,\U_i$ 
respectively.
\item Delete those edges in $E_j, j>i$, which are incident on nodes in $\odd_i
\cup\U_i$. These are the nodes that are matched by every maximum matching in 
$G'_i$.
\item Delete all the edges from $G'_i$ between a node in 
$\odd_i$ and 
a node in $\odd_i\cup\U_i$. We refer to these edges as $\odd_i\odd_i$ and
$\odd_i\U_i$ edges respectively. These are the edges which do not belong to any
maximum matching in $G'_i$.
\item Add the edges in $E_{i+1}$ to $G'_i$ and call
the resulting graph $G'_{i+1}$.
\item Determine a maximum matching $M_{i+1}$ in $G'_{i+1}$ by augmenting $M_i$.
\end{enumerate}}
}

The algorithm constructs a graph $G'_{r+1}$. We construct a {\em reduced graph}
$G'$ by deleting all the edges from $G'_{r+1}$
between a node in $\odd_{r+1}$ and a node in $\odd_{r+1}\cup\U_{r+1}$.
The graph $G'$ will be used in subsequent sections.

We note the following invariants of Irving~et~al.'s algorithm:
\begin{enumerate}
\item [($I1$)] For every $1 \le i \le r$, every rank-maximal matching in $G_i$ 
is contained in $G'_i$.
\item[($I2$)] The matching $M_i$ is rank-maximal in $G_i$, and is a maximum 
matching in $G'_i$.
\item [($I3$)] If a rank-maximal matching in $G$ has signature $(s_1,\ldots,s_i,\ldots 
s_r)$ then $M_i$ has signature $(s_1,\ldots,s_i)$.
\item [($I4$)]The graphs $G'_i$, $1\leq i\leq r+1$
constructed at the end of iteration $i$ of Irving~et~ al.'s algorithm, 
and $G'$ are independent of 
the rank-maximal matching computed by the algorithm. This follows from 
Lemma \ref{lem:node-class} and invariant $I2$. 
\end{enumerate}


\section{Switching Graph Characterization}\label{sec:switch}
In this section, we describe the switching graph characterization of 
rank-maximal matchings, 
and show its application in computing {\em rank-maximal
pairs}. 
   
Let $M$ be a rank-maximal matching in $G$ and let $G'=(\A \cup \p, E')$ be
the reduced graph as described in Section \ref{sec:prelim}. 

\begin{definition}[Switching Graph]\label{def:switch-graph}
The switching graph $G_M=(V_M,E_M)$ with respect to a rank-maximal matching $M$
is a directed weighted graph
with $V_M=\p$ and $E_M=\{(p_i,p_j)\mid \exists a\in \A, (a,p_i)\in M, 
\textrm(a,p_j)\in E'\}$. Further, weight of an edge $(p_i,p_j)$ is
$w(p_i,p_j)=rank(a,p_j)-rank(a,p_i)$, where $rank(a,p)$ is the 
rank of a post $p$ in the preference list of an applicant $a$. 
\end{definition}

Thus an edge $(p_i, p_j) \in E_M$
iff there exists an applicant $a$ such that $(a, p_i) \in M$ and $(a, p_j)$
is an edge in the graph $G'$. 
We define the following notation:
\begin{enumerate}
\item {\em Sink vertex:} 
A vertex $p$ of $G_M$ is called a {\em sink} vertex, if $p$
has no outgoing edge in $G_M$ and
$p \in \EE_1 \cap \EE_2 \cap \ldots \cap \EE_{r+1}$. Recall that 
$\EE_i$ is the set of vertices which were even in the graph $G'_i$ constructed
in the $i$th iteration of Irving~et~al.'s algorithm.
\item {\em Sink and non-sink components of $G_M$:}
A connected component $\X$ in the underlying undirected graph of $G_M$ is
called a {\em sink component} if $\X$ contains one or more
sink vertices, and a {\em non-sink component} otherwise.
\item {\em Switching paths and switching cycles: }
A path $T = \langle p_0, p_1 \ldots, p_{k-1}\rangle$ in $G_M$
is called a {\em switching path} if $T$ ends in a sink vertex and $w(T) = 0$.
Here, $w(T)$ is the sum of the weights of the edges in $T$.
A cycle $C = \langle p_0,  \ldots, p_{k-1}, p_0\rangle$ in $G_M$ is called a 
{\em switching cycle} if $w(C) = 0$.
\item {\em Switching operation:} 
Let $T=\langle p_0, p_1 \ldots, p_{k-1}\rangle$ be a switching path in $G_M$. 
Let $\A_T=\{a\in \A \mid M(a)\in T\}$. Further, let $M(a_i)=p_i$ for $0\leq i
\leq k-2$. We denote by 
$M'=M\cdot T$, the matching obtained by {\em applying} $T$ to $M$. Thus, for
$a_i\in \A_T$, $M'(a_i)=p_{i+1}$, and for $a\notin \A_T$, $M'(a)=M(a)$.
The matching $M\cdot C$, obtained by {\em applying} a switching cycle $C$ to $M$
is defined analogously. 
We also refer to $M\cdot C$ or $M\cdot T$ as a {\em 
switching operation}.
\end{enumerate}
Figure~\ref{fig:example-ex2} illustrates an example instance along with its switching graph.
\vspace{-0.2in}
\begin{figure}[hbt]
\begin{minipage}[b]{0.45\linewidth}
\centering
\begin{equation*}
\setlength{\arraycolsep}{0.5ex}\setlength{\extrarowheight}{0.25ex}
\begin{array}{@{\hspace{1ex}}c@{\hspace{1ex}}
              @{\hspace{1ex}}c@{\hspace{1ex}}
              @{\hspace{1ex}}c@{\hspace{1ex}}
              @{\hspace{1ex}}c@{\hspace{1ex}}
              @{\hspace{1ex}}c@{\hspace{1ex}}
              @{\hspace{1ex}}c@{\hspace{1ex}}}
    a_{1}: & { p_1} \ & p_2 \ & p_3   \\[.5ex] 
    a_{2}: & {p_1} \ & p_2 \ & p_4  \ &\\[.5ex] 
    a_{3}: & p_1 \\ [.5ex] 
    a_{4}: & p_5 \ & p_6 \ & p_7 \\[.5ex] 
    a_{5}: & p_5 \ & p_6 \ & p_7 \\[.5ex] 
    a_{6}: & p_5 \ & p_6 \ & p_7 \\[.5ex] 
\end{array}
\end{equation*}
(a)
\end{minipage}
\begin{minipage}[b]{0.45\linewidth}
\centering

\psfrag{p1}{$p_1$}
\psfrag{p2}{$p_2$}
\psfrag{p3}{$p_3$}
\psfrag{p4}{$p_4$}
\psfrag{p5}{$p_5$}
\psfrag{p6}{$p_6$}
\psfrag{p7}{$p_7$}
\psfrag{p8}{$p_8$}
\psfrag{p9}{$p_9$}

\psfrag{-1}{$-1$}
\psfrag{1}{$1$}
\psfrag{2}{$2$}
\psfrag{-2}{$-2$}
\includegraphics[scale=0.4] {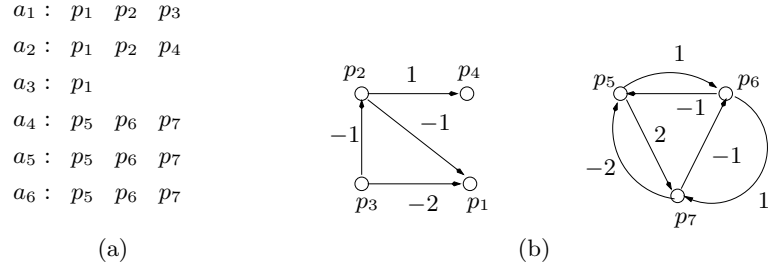}
(b)
\end{minipage}
\caption{(a) Preference lists of agents $\{a_1, \ldots, a_6\}$ in increasing 
order of ranks.
(b) Switching graph $G_M$ with respect to rank-maximal matching $M=\{(a_1,p_3),
(a_2,p_2),(a_3,p_1),(a_4,p_7),(a_5,p_5),(a_6,p_6)\}$. The vertex $p_4$ is the
only sink-vertex
and the path $(p_3, p_2, p_4)$ is a switching path. Note that every directed 
cycle is a switching cycle.}
\label{fig:example-ex2}
\end{figure}
\subsection{Properties of the switching graph}\label{sec:switch-prop}
In this section, we prove several useful properties of the switching
graph by characterizing switching paths and switching cycles. 

In the following lemma, we show that a switching operation on a rank-maximal 
matching $M$ results in another rank-maximal matching in $G$. 

\begin{lemma}\label{lem:0wt-path}
Let $T$ (resp. $C$) be a switching path (resp. switching cycle) in $G_M$. 
Then, $M' = M \cdot T$ (resp. $M' = M \cdot C$) is a rank-maximal matching in 
$G$.
\end{lemma}
\begin{proof}
We prove the lemma for a switching path $T$. A similar argument follows for a 
switching cycle.
To show that $M'$ is rank-maximal, we show that $M$ and $M'$ have the same 
signature. 

Let $T= \langle p_0, p_1, \ldots, p_{k-1} \rangle$ be a switching path in $G_M$.
Let $\A_T =  \{a \mid M(a)\in T\}$. 
By the definition of a {\em switch}, we know that
$|M| = |M'|$ and for each $a \notin \A_T$, we have $M'(a) = M(a)$. Thus,
it suffices to show that the signatures of $M$ and $M'$ restricted to the 
applicants in 
$\A_T$ are the same. We denote them by $\rho_{T}(M)=(x_1,x_2,\ldots,x_r)$ 
and $\rho_T(M')=(y_1,y_2,\ldots,y_r)$ respectively.
Note that an edge of rank~$i$ in $M$ contributes $-i$ to the weight 
of $T$, whereas one in $M'$
contributes $i$. Further, since $T$ is a switching path, $w(T) = 0$. Thus, 
\begin{equation}
w(T) = (y_1-x_1)+2(y_2-x_2)+\ldots+r(y_r-x_r)=0 \label{eqn1tmp1}
\end{equation}
Since we consider only applicants in $\A_T$, we know that,
$\sum_{i=1}^r x_{i}  =  \sum_{i=1}^r y_i$, i.e.,
\begin{equation}
 \sum_{i=1}^r (x_i  -  y_i)  =  0\label{eqn2tmp2}
\end{equation}
For contradiction, assume that $\rho_{T}(M) \succ \rho_{T}(M')$. That is, 
there exists an index $j$ such that $x_j>y_j$ and, for $1 \le i < j$,
we have $x_i = y_i$.
Then, for 
Eqn.~\ref{eqn2tmp2} to be satisfied, there exists an index $\ell>j$ 
such that 
$x_{\ell} < y_{\ell}$.
In fact we will show the following stronger claim:

\begin{claim}\label{clm:clm1}
There exists an index $\ell > j$ such that $\sum_{i=1}^{\ell} (x_i - y_i)  < 0$.
\end{claim}

Before proving the claim, we show how it suffices to complete the 
proof of the lemma. 
Assuming the claim, consider the reduced graph $G'_{\ell}$ constructed in
the $\ell$th iteration of Irving~et~al.'s algorithm.

As $\sum_{i=1}^{\ell}(x_i - y_i)
<0$, we have $\sum_{i=1}^{\ell} x_{i}< \sum_{i=1}^{\ell} y_{i}$. Thus 
$|M\cap G'_{\ell}|<|M'\cap G'_{\ell}|$. However, by Invariant $(I2)$ (ref. 
Section~\ref{sec:prelim}), this contradicts the fact that every 
rank-maximal
matching restricted to any rank $\ell$ is also a maximum matching in the reduced
graph $G'_{\ell}$. This completes the proof of the lemma. We prove the claim 
below. 

{\it Proof (of claim)}:
Assume the contrary, i.e. $\sum_{i=1}^{k}(x_i - y_i)\geq 0$ for all $k$. Note
that this is trivially true for $k\leq j$, by our choice of $j$. 
Equivalently, $\sum_{i = k+1}^r (x_i - y_i)\leq 0$ for all $k$. Define
$T_k=\sum_{i=k}^r(x_i-y_i)$ for $1\leq k\leq r$. Thus, to prove the claim,
it suffices to show that there exists an index $\ell$ such that $T_\ell>0$.
Now consider Eqn.~\ref{eqn1tmp1}. It can be rewritten as follows:
\begin{equation}\label{eqn:split}
(x_1 - y_1)+2(x_2 - y_2)+\ldots+r(x_r - y_r)= T_1+T_2+\ldots+T_r=0
\end{equation}
We know that $T_1=0$, because it is the left-side of Eqn.~\ref{eqn2tmp2}.
Now, consider the term $T_r=x_r-y_r$. If $T_r=0$, we can eliminate $x_r$ and 
$y_r$ and get equations in $r-1$ variables. 
If $T_r > 0$, then 
Eqn.~\ref{eqn2tmp2} implies that the claim holds for $k=r-1$. So, without loss
of generality, we can assume $T_r<0$.
But then, to satisfy Eqn.~\ref{eqn:split}, there exists an index $i$, 
$1<i<r$, such that $T_i>0$. This implies that the claim holds for $\ell=i-1$.
This completes the proof of the claim.
\qed
\end{proof}
Now we address the question of recognition of switching
paths and switching cycles in $G_M$.
In Lemma~\ref{lem:rmm-cycle}, we show that every cycle in $G_M$ is in fact a 
switching cycle, that is, a zero-weight cycle. In 
Lemma~\ref{lem:switch-path}, we characterize switching paths. 
\begin{lemma}
\label{lem:rmm-cycle}
Let $M$ be a \rmm in $G$, and $C$ be a cycle in $G_M$. Then $w(C) = 0$.
\end{lemma}
\begin{proof}(Sketch) 
Let $C'$ be the alternating
cycle in $G'$, corresponding to the cycle $C$ in $G_M$.
To prove the Lemma, it suffices to show that, $C'$ has an 
equal number 
of matched and unmatched edges of any rank $i$, and hence $w(C)=0$. 
We prove this by induction on $i$.
See Appendix~\ref{sec:app-switch} 
for details.
\qed
\end{proof}

\begin{lemma}\label{lem:switch-path}
Let $M$ be a \rmm in $G$, and $G_M$ be the switching graph
with respect to $M$. Recall that $\even_i$ is the set of even vertices in the 
graph $G'_i$ constructed in the $i$th iteration of Irving~et~al.'s algorithm.
The following properties hold :
\begin{enumerate}
\item Let $p$ be an unmatched post in $M$. Then $p \in \even_1 \cap  
\ldots \cap \even_{r+1}$ and therefore is a sink in $G_M$.
\item A post $p$ belongs to a \Ecomp iff $p \in \even_{r+1}$. A post $p$ 
belongs to a \Ucomp iff $p \in \un_{r+1}$.
\item Let $T$ be a path from a post $p$ to some sink $q$ in $G_M$. Then $w(T)=0$
iff $p\in \EE_1\cap\ldots\cap\EE_{r+1}$.
%
%
\end{enumerate}
\end{lemma}

The proof appears in Appendix~\ref{sec:app-switch}.
In the following theorem,
we prove that every rank-maximal matching can be obtained from $M$ by applying
suitable {\em switches}.
We include the proof in Appendix~\ref{sec:app-switch}.
\begin{theorem}\label{thm:another-rmm}
Every rank-maximal matching $M'$ in $G$ can be obtained from $M$ by applying to $M$ 
vertex-disjoint switching paths and switching cycles in $G_M$.
\end{theorem}
\vspace{-0.2in}
\subsection{Generating all rank-maximal pairs}\label{sec:rmm-pairs}
\vspace{-0.1in}
In this section we give an efficient algorithm to compute the set of
rank-maximal pairs, defined below: 
\begin{definition}\label{def:rmm-pair}
An edge $(a, p)$ is a rank-maximal pair if
there exists a rank-maximal matching $M$ in $G$ such that $M(a) = p$.
\end{definition}
We refer to rank-maximal pairs as {\em rmm-pairs}.
We show that the set of rmm-pairs can be computed in time linear in the size 
of the switching graph 
$G_M$ constructed with respect to any rank-maximal matching $M$.
We prove the following theorem:

\begin{theorem}\label{thm:rmm-pairs}
The set of rmm-pairs for an instance $G= (\A \cup \p, E)$  can be computed in 
$O(\min(n+r,r\sqrt{n})m)$ time.
\vspace{-0.1in}
\end{theorem}
\begin{proof}
We note that, by Theorem~\ref{thm:another-rmm}, an edge $(a,p)$ is a rmm-pair
{\em iff}
(i) $(a,p)\in M$ or,
 (ii) the edge $(M(a),p)$ belongs to a switching cycle in $G_M$ or,
(iii) the edge $(M(a),p)$ belongs to a switching path in $G_M$.

Condition~(i) can be checked by computing a rank-maximal matching $M$ which takes $O(\min(n+r,r\sqrt{n})m)$ time.
Condition~(ii) can be checked by computing
strongly connected components of $G_M$, which takes time linear in the size of
$G_M$.

To check Condition~(iii), note that a post $p$ has
a zero-weight path to a sink if and only if $p\in \EE_1\cap\ldots\cap
\EE_{r+1}$ by Lemma~\ref{lem:switch-path}~(3).
Moreover, all the paths from such a post $p$ to a sink have weight zero.
Therefore, performing a DFS from each $p\in\EE_1\cap\ldots\cap\EE_{r+1}$
and marking all the edges encountered in the DFS (not just the tree edges) gives all the pairs which
satisfy Condition~(iii).
\qed
\end{proof}

\section{Counting Rank-Maximal Matchings}\label{sec:count}
\vspace{-0.1in}
We prove that the problem of counting the number of 
rank-maximal matchings in an instance is \shP-complete,
and give an \fpras\ for the same.


\vspace{-0.1in}
\subsection{Hardness of Counting}\label{sec:shp-hard}
We prove \shP-hardness by reducing the problem of counting the number of 
matchings in $3$-regular bipartite graphs to counting the number of rank-maximal
matchings. The former was shown to be \shP-complete by Dagum and Luby 
\cite{DL92}.

\noindent{\bf Reduction for lists with ties: }
First let us consider the case when preference lists may contain 
ties\footnote{ Recall that preference lists are said to contain ties if an 
applicant ranks two or more posts at the same rank.}. 
Let $H = (X \cup Y,E )$ be a $3$-regular bipartite graph. We construct an 
instance $G$ of the rank-maximal
matchings problem by setting $G = H$ and assigning rank~$1$ to all the edges in 
$E$. It is well-known that a $k$-regular bipartite
graph admits a perfect matching for any $k$. It is easy to see that every 
perfect matching in $H$ is a rank-maximal 
matching in $G$ and vice versa. 
This proves the \shP-hardness of the problem for the case of ties.
 
\noindent{\bf Reduction for strict lists: }
Let $H = (X \cup Y, E)$ be a $3$-regular bipartite graph, with $|X| = |Y| = n$. 
The corresponding instance $G = (\A \cup \p, E_G)$ of the rank-maximal 
matchings problem is as follows: 
\begin{displaymath}
\A = \{ a_x: x \in X \} \cup \{ ad_1, ad_2, \ldots, ad_{n-3} \};
\p = \{ p_y: y \in Y \} \cup \{ pd_1, pd_2, \ldots, pd_{n-3} \} 
\end{displaymath}
Here $ad_i,pd_i$, $1\leq i\leq n-3$ are dummy agents and 
dummy posts respectively.

To construct the preference lists of agents in $\A$,
we fix an arbitrary ordering on the vertices in $Y$ i.e. 
$order: Y \rightarrow \{1, \ldots, n \}$. This assigns an ordering on the posts in $\p$.
The preference lists of the agents can be described as below:
\begin{itemize}
\item A dummy agent $ad_i$ has a preference list 
of length one, with dummy post $pd_i$ as his rank~$1$ post. 
\item 
The preference list of an agent $a_x$ consists of posts $p_{y_1}, p_{y_2}, 
p_{y_3}$ ranked at $order(y_1), order(y_2)$, and 
$order(y_3)$
respectively, where 
$y_1, y_2, y_3$ denote the 3 neighbors of $x$ in $H$. 
The remaining places in the preference list of $a_x$ are filled using the $n-3$ 
dummy posts. 
\end{itemize}
Following Lemma (proof in Appendix~\ref{sec:app-count}) shows the correctness of the reduction.
\begin{lemma}\label{lem:pm-vs-rmm}
Let $H$ be a 3-regular bipartite graph and let $G$ be the rank-maximal
matchings instance constructed from $H$ as above. There is a one-to-one correspondence between perfect matchings in $H$ and rank-maximal matchings in $G$.
\end{lemma}

Using Lemma~\ref{lem:pm-vs-rmm} and our observation for ties, we conclude the following:
\begin{theorem}
The problem of counting the number of rank-maximal matchings in an instance is
\shP-Complete for both
strict and tied preference lists.
\end{theorem}
\subsection{An \fpras\ for Counting Rank-Maximal Matchings}
\label{sec:fpras}
Given the hardness result in Section~\ref{sec:shp-hard}, 
it is unlikely to be able to
count the number
of rank-maximal matchings in an instance in polynomial time. We now show that
there exists a fully polynomial-time randomized 
approximation scheme (\fpras) for the problem. 
We use the following result by Jerrum~et~al.
\cite{JerrumSV04}:
\begin{theorem}[\cite{JerrumSV04}]
\label{thm:jerrum}
There exists an \fpras\ for the problem of counting the number of perfect 
matchings in a bipartite graph.
\end{theorem} 
We give a polynomial-time reduction from the problem of counting the number
of rank-maximal matchings (denoted as \shRMM) to the problem of counting the number of
perfect matchings in a bipartite graph (denoted as \shBPM). 

\noindent {\bf Reduction from \shRMM\ to \shBPM:}
Given an instance $G=(\A\cup\p,E)$ of the rank-maximal matchings problem, we
first construct another instance $H$ of the rank-maximal matchings problem,
which is used to get an instance $I$ of the bipartite perfect 
matchings problem. The steps of the construction are as follows: 

\begin{enumerate}
\item For every  $a \in \A$, introduce a dummy last-resort post $\ell(a)$ ranked  $r+1$. This ensures that every rank-maximal matching is $\A$-complete.
\item Let $M$ be any rank-maximal matching in $G$, let $G'$ be the reduced 
graph
obtained by Irving~et~al.'s algorithm (ref. Section~\ref{sec:prelim}).
\item Let $k$ be the number of unmatched posts in $G'$. 
Introduce $k$ dummy applicants
$ad_1,\ldots,ad_k$. The preference list of each dummy applicant consists of
all the posts in $G'$ which are in $\EE_1\cap\ldots\cap\EE_{r+1}$, tied at rank
$r+2$.
\item The instance $H$ consists of all the applicants in $G$ and their 
preference lists in $G$, together with the
dummy applicants and their preference lists introduced above. The set of posts
in $H$ is the same as that in $G$.
\item The instance $I$ of bipartite perfect matchings problem is simply the
reduced graph $H'$, obtained by executing Irving et al.'s algorithm on $H$.
\end{enumerate}
Correctness of the reduction follows from the following lemma, the proof (in Appendix~\ref{sec:app-count}) uses
the switching graph characterization.

\begin{lemma}\label{lem:fpras-correct}
Let $G$ be the rank-maximal matchings instance and let $H$  and $I$ be the rank-maximal matchings
instance and the  bipartite perfect matchings instance respectively constructed as above. Then, the following hold:\\
$1.$ Corresponding to each rank-maximal matching $M$ in $G$, 
there are exactly $k!$ distinct rank-maximal matchings in $H$.\\
$2.$ Each rank-maximal matching in $H$ matches all the applicants and posts, and
all its edges appear in $I$. Hence it is a perfect matching in the instance $I$.\\
$3.$ A matching in $G$ that is not rank-maximal has no
corresponding perfect matching in $I$.
\end{lemma}
The \fpras\ for \shRMM\ involves the following steps:
\begin{enumerate}
\item The reduction from \shRMM\ instance $G$ to \shBPM\ instance $I$, 
\item Running Jerrum et al.'s \fpras\ on $I$ to get an 
approximate count, say $C$, of the number of perfect matchings in $I$,
\item Dividing $C$ by $k!$ to get an approximate count of number of rank-maximal
matchings in $G$.
\end{enumerate}
Steps $1$ and $2$ clearly work in polynomial time. 
For step $3$, note that both $C$ and $k$ are at most $n!$ and can be represented
in $O(n\log n)$ bits, which is
polynomial in the size of $G$. Therefore Step~$3$ also
works in polynomial time. This completes the \fpras\ for \shRMM\ problem. 

\section{Popularity of Rank-Maximal Matchings}\label{sec:popular}
As mentioned earlier, an instance of the rank-maximal matchings
problem may admit more than one rank-maximal matching. To choose
one rank-maximal matching, it is natural to impose an additional optimality criterion.
Such a question has been considered earlier in the context of popular matchings by \cite{KN08,MI11}
and also in the context of the stable marriage problem \cite{ILG87}.
The additional notion of optimality that we impose is the notion of
popularity, defined below: 

\begin{definition}[Popular matching]\label{def:pop}
A matching $M$ is {\em more popular  than} matching $M'$ (denoted by $M \succ_p M'$) if the number of applicants that
prefer $M$ to $M'$ is more than the number of applicants that prefer $M'$ to $M$. A matching
$M$ is popular if no matching $M'$ is more popular than $M$.
\end{definition}
An applicant $a$ prefers matching $M$ to $M'$ if either
(i) $a$ is matched in $M$ and unmatched in $M'$, or (ii) $a$ is matched in both and prefers the post $M(a)$ 
to $M'(a)$. We consider the following question: Given an instance of the 
rank-maximal matchings problem,
is there a rank-maximal matching that is popular in the set of all 
rank-maximal matchings? 
We refer to such
a matching as a {\em popular rank-maximal matching}.
There are simple instances in which there is no popular
matching; further there is no popular rank-maximal matching. 
However, if a popular rank-maximal matching exists, it 
seems an appealing choice
since it enjoys both rank-maximality and popularity.
We make partial progress on this question. Using
the switching graph characterization developed in Section~\ref{sec:switch},
we give a simple algorithm to determine if a given rank-maximal matching 
$M$ is a popular rank-maximal matching. If not, our algorithm
outputs a rank-maximal matching $M'$ which is more popular than $M$.
\paragraph{Outline of the algorithm:}
Given a graph $G = (\A \cup \p, E)$ and a rank-maximal matching $M$ in $G$, the
 algorithm first constructs the switching graph $G_M$
corresponding to $M$.
Now consider the following re-weighted graph $\tilde{G}_M$ 
where positive
weights of edges in $G_M$ are replaced by $+1$ weights 
and negative weights by $-1$. 
Thus a $-1$ weight edge $(p_i, p_j)$ in $\tilde{G}_M$ implies
that $M(p_i)$ prefers $p_j$ to $p_i$. 

Let $T$ be a switching path in $G_M$, and let $\tilde{T}$ be the corresponding 
path in $\tilde{G}_M$. 
It is easy to see that if $w(\tilde{T}) < 0$ in $\tilde{G}_M$,
then  $M' = M \cdot T$ is
more popular than $M$. Same holds for a switching cycle in $G_M$. 
Therefore, $M$ is a popular rank-maximal matching, 
if and only if there is no negative-weight path to sink
or negative-weight cycle in $\tilde{G}_M$. 

To check this, we use shortest path 
computations using Bellman-Ford algorithm in a suitably modified graph. 
The
details of the algorithm and proof of the following lemma, which establishes 
correctness, appear in Appendix~\ref{sec:app-popular}.
\begin{lemma}\label{lem:popular}
A given rank-maximal matching $M$ is popular if and only if there is no 
negative-weight path to a sink or a negative-weight cycle in the re-weighted 
switching graph.
\end{lemma}
Thus we get an $O(mn)$ time algorithm for checking whether a given rank-maximal
matching is a popular rank-maximal matching, where $m$ and $n$ are number of
edges and vertices in the switching graph respectively.

\REM {
\section{Conclusion}
}

\noindent{\bf Acknowledgment:} We thank Partha Mukhopadhyay for a proof of Lemma \ref{lem:0wt-path}.

\bibliographystyle{abbrv}
\bibliography{references}

\begin{thebibliography}{10}

\bibitem{ACMM04}
D.~J. Abraham, K.~Cechl\'arov\'a, D.~F. Manlove, and K.~Mehlhorn.
\newblock Pareto-optimality in house allocation problems.
\newblock In {\em Proceedings of 15th ISAAC}, pages 3--15, 2004.

\bibitem{AIKM07}
D.~J. Abraham, R.~W. Irving, T.~Kavitha, and K.~Mehlhorn.
\newblock Popular matchings.
\newblock {\em SIAM Journal on Computing}, 37(4):1030--1045, 2007.

\bibitem{DL92}
P.~Dagum and M.~Luby.
\newblock Approximating the permanent of graphs with large factors.
\newblock {\em Theor. Comput. Sci.}, 102(2):283--305, 1992.

\bibitem{HZ79}
A.~Hylland and R.~Zeckhauser.
\newblock The efficient allocation of individuals to positions.
\newblock {\em Journal of Political Economy}, 87(2):293--314, 1979.

\bibitem{Irving03}
R.~W. Irving.
\newblock Greedy matchings.
\newblock {\em Technical Report, University of Glasgow}, TR-2003-136, 2003.

\bibitem{IKMMP06}
R.~W. Irving, T.~Kavitha, K.~Mehlhorn, D.~Michail, and K.~E. Paluch.
\newblock Rank-maximal matchings.
\newblock {\em ACM Transactions on Algorithms}, 2(4):602--610, 2006.

\bibitem{ILG87}
R.~W. Irving, P.~Leather, and D.~Gusfield.
\newblock An efficient algorithm for the ``optimal'' stable marriage.
\newblock {\em Journal of the ACM}, 34(3):532--543, 1987.

\bibitem{JerrumSV04}
M.~Jerrum, A.~Sinclair, and E.~Vigoda.
\newblock A polynomial-time approximation algorithm for the permanent of a
  matrix with nonnegative entries.
\newblock {\em J. ACM}, 51(4):671--697, 2004.

\bibitem{KN08}
T.~Kavitha and M.~Nasre.
\newblock Note: Optimal popular matchings.
\newblock {\em Discrete Applied Mathematics}, 157(14):3181--3186, 2009.

\bibitem{KS06}
T.~Kavitha and C.~D. Shah.
\newblock Efficient algorithms for weighted rank-maximal matchings and related
  problems.
\newblock In {\em Proceedings of 17th ISAAC}, pages 153--162, 2006.

\bibitem{MI11}
E.~McDermid and R.~W. Irving.
\newblock Popular matchings: structure and algorithms.
\newblock {\em Journal of Combinatorial Optimization}, 22(3):339--358, 2011.

\bibitem{Nasre13}
M.~Nasre.
\newblock {Popular Matchings: Structure and Cheating Strategies}.
\newblock In {\em Proceedings of 30th STACS}, pages 412--423, 2013.

\bibitem{Paluch13}
K.~E. Paluch.
\newblock Capacitated rank-maximal matchings.
\newblock In {\em Proceedings of 8th CIAC}, pages 324--335, 2013.

\bibitem{GGL95new}
W.~R. Pulleyblank.
\newblock Handbook of combinatorics (vol. 1).
\newblock chapter Matchings and Extensions, pages 179--232. MIT Press,
  Cambridge, MA, USA, 1995.

\bibitem{Yuan96}
Y.~Yuan.
\newblock Residence exchange wanted: A stable residence exchange problem.
\newblock {\em European Journal of Operational Research}, 90(3):536 -- 546,
  1996.

\end{thebibliography}
\newpage
\begin{appendix}
\section{Details from Section~\ref{sec:switch}}\label{sec:app-switch}
\label{app:sec-switch}

\REM{\noindent {\bf Example instance:} Consider an instance with $\A = \{a_1, \ldots, a_6\}$ and $\p = \{p_1, \ldots, p_7\}$. The
preference lists of the applicants are as given in Figure~\ref{fig:example-ex2}(a). The preference lists
can be read as follows: applicant $a_1$ treats posts $p_1$, $p_2$, and $p_3$ as his rank~$1$, rank~$2$, and rank~$3$
posts respectively. Consider the following rank-maximal matching $M = \{(a_1, p_3), (a_2, p_2), (a_3, p_1), (a_4, p_5), (a_5, p_6), (a_6, p_7)\}$.
Figure~\ref{fig:example-ex2}(b) shows the corresponding switching graph $G_M$. The switching graph has two components,
the one containing $p_4$ is a sink component, whereas the other containing $p_5$ is the non-sink component. We remark that
both $p_1$ and $p_4$ are vertices with no out-going edges, however, only $p_4$ is a sink vertex since it is also unmatched in $M$.
The path $ T = \langle p_3, p_2, p_4\rangle$ is a zero weight path ending in a sink and therefore a switching path in $G_M$.
Thus, applying $T$ to $M$, we get another matching $M' = M \cdot T$ which is also a rank maximal matching in the instance.
The matching $M' = \{(a_1, p_2), (a_2, p_4), (a_3, p_1), (a_4, p_5), (a_5, p_6), (a_6, p_7)\}$.  It is easy to verify that
applying any cycle $C$ (which is in fact a zero weight cycle) to the matching $M$ also leads to a rank-maximal matching in the instance.}

\begin{appendix-lemma}{\ref{lem:rmm-cycle}}
Let $M$ be a \rmm in $G$ and $C$ be a cycle in $G_M$, then
$w(C) = 0$.
\end{appendix-lemma}
\begin{proof}
Let $C$ be a cycle in $G_M$ and let $C'$ denote the corresponding alternating
cycle in $G'$. To prove the Lemma statement, we show that, for any rank $i$, $C'$ has an equal number
of matched and unmatched edges, and hence $w(C)=0$.  We use induction on $i$ to prove that
for any $i$, the cycle $C'$ has equal number of matched and unmatched edges.  
Let us partition the edges of $C'$ as $X_1 \cup  \ldots \cup X_r$, 
where $X_i$ denotes edges of rank~$i$ belonging to $C'$. Note that, for some 
$i$,  $X_i$ may be empty. Now for any $i$, consider $Y_i = \cup_{j=1}^{i}
X_j$. We show by induction on $i$, that for each $i$, any component of $Y_i$ is 
either an even length path or the cycle $C'$ itself.

{\bf Base case:} Let $\ell$ denote the first index for which $X_{\ell}$ is 
non-empty. 
Then each $j<\ell$ trivially satisfies the induction hypothesis.
If $Y_{\ell} = C'$, we are done, since $C'$ is an alternating cycle,
with equal
number of unmatched and matched edges, 
all of rank~$\ell$. 
If $Y_{\ell} \neq C'$, 
then for contradiction,
let $Y_{\ell}$ contain an odd length alternating path 
$T = \langle a_1, p_1, a_2, p_2, \ldots, a_k, p_k \rangle$.
Since all the other edges in $C'$ are of rank greater than $\ell$, and they are
incident on $a_1$ and $p_k$, both $a_1$ and $p_k$ must belong to $\even_{\ell}$ in $G'_{\ell}$
at the end of the $\ell$th iteration of Irving et al.'s algorithm. 
However, since $T$ is present in
$G'$
it must be present in $G'_{\ell}$.
Note that
both $a_1$ and  $p_k$ belong to $\even_{\ell}$ and the path $T$  is an alternating path of odd length.
Now consider  labeling the vertices of $T$ as $\odd_{\ell}$ or $\even_{\ell}$ from both $a_1$ and $p_k$.
It is clear that we either encounter an
$\odd_{\ell} \odd_{\ell}$ edge which must have been deleted in Irving et al.'s 
algorithm or an $\even_{\ell} \even_{\ell}$ edge
which cannot be present in $G$ (by Lemma~\ref{lem:node-class} (c)).
Thus in either case, we get a contradiction. 
Hence $T$ must be an even length path,
with equal number of matched and unmatched rank~$\ell$ edges.

{\bf Induction step:}
Assuming the induction hypothesis for some $\ell<r$, the proof for $Y_{\ell+1}$
is similar to that of base case.

The above implies that, for every $X_i$, $C'$ has $|X_i|/2$ matched and unmatched
edges. Hence the corresponding cycle $C$ in $G_{M}$ has zero weight,
and $M\cdot C$ is a rank-maximal matching.
\qed
\end{proof}

\begin{appendix-lemma}{\ref{lem:switch-path}}
Let $M$ be a \rmm in $G$ and $G_M$ be the switching graph
with respect to $M$. The following properties hold:
\begin{enumerate}
\item\label{itm:zero} Let $p$ be an unmatched post in $M$. Then $p \in \even_1 \cap  
\ldots \cap \even_{r+1}$ and therefore  is a sink in $G_M$.
\item\label{itm:comp} A post $p$ belongs to a \Ecomp iff $p \in \even_{r+1}$. A post $p$ 
belongs to a \Ucomp iff $p \in \un_{r+1}$.
%
%
\item\label{itm:0wtpath} Let $T$ be a path from a post $p$ to some sink $q$ in
$G_M$. Then $w(T)=0$ iff $p\in \EE_1 \cap \ldots \cap \EE_{r+1}$.
\end{enumerate}
\end{appendix-lemma}
\begin{proof}
\begin{enumerate}
\item 
The proof follows by observing that every rank-maximal matching keeps vertices
 in $\odd_i \cup \un_i$ matched
for every $i = 1 \ldots r+1$. Thus if $p$ is unmatched in a rank-maximal 
matching $M$, then $p \in \even_1 \cup \ldots \cup \even_{r+1}$.
\item Consider a post $p \in \even_{r+1}$. If $p$ is unmatched in $M$, then by 
\ref{itm:zero} above,
$p$ is a sink vertex and therefore belongs to a sink component. 
Now, assume that $p$ is matched but
since it belongs to $\even_{r+1}$, $p$ has an even length alternating path starting at
an unmatched node $p'$ with respect to $M$ in $G'$. Let the alternating path be denoted by $\langle p = p_1, a_1, \ldots, p_k, a_k, p_{k+1} = p'\rangle$.
Note that for every $i=1, \ldots k$, we have $M(a_i)= p_i$. Further, every unmatched edge $(a_i, p_{i+1})$ is
of the form $\odd_{r+1} \even_{r+1}$. Therefore no such unmatched edge gets 
deleted in the $(r+1)$st iteration of Irving et al.'s algorithm. 
This implies that the directed path $\langle p = p_1, p_2, \ldots, p_{k+1} =p'\rangle$ is present in $G_M$.
Thus, $p$ belongs to the \Ecomp that contains $p'$.

To prove the other direction
let $\X$ be a \Ecomp in $G_M$ and $p'$ be a sink in $\X$. For the sake of 
contradiction, let $p' \in \X$ and $p \in \un_{r+1}$.
Recall that $\odd_{r+1} \cap \p = \emptyset$.

Now since $p$ and $p'$ lie in the same component, there is an  (undirected) path
between $p$ and $p'$ in the underlying undirected component of $\X$.
Let $\langle p= p_1, p_2, \ldots, p_k= p'\rangle$ denote this undirected path.
Since $p_1 \in \un_{r+1}$ and $p_k \in \even_{r+1}$, it implies that there 
exists an index $i$, $1\le i \le k-1$,
such that $p_i \in \un_{r+1}$ and $p_{i+1} \in \even_{r+1}$. 
Note that,
by the above argument, $p_{i+1}$ has a directed path $T$ to some sink $q$ 
in $\X$.

Consider the two possible directions
 for the edge between $p_i$ and $p_{i+1}$ in $G_M$:
\begin{enumerate}
\item If the edge is directed from $p_i$ to $p_{i+1}$ in $G_M$,
then 
the path $T$ from $p_{i+1}$ to the sink $q$ in $\X$ can be prefixed with 
the edge $(p_i, p_{i+1})$ to get a directed path from $p_i$
to $q$. This implies that there is an even-length alternating path
with respect to $M$ in $G'$ from $q$ to $p_i$.
This contradicts the fact that $p_i \in \un_{r+1}$. 

\item Finally, if the edge is directed from $p_{i+1}$ to
$p_i$ in $G_M$, then it implies that $p_{i+1}$ is matched in $M$ and let 
$M(p_{i+1}) = a_{i+1}$. Since $p_{i+1} \in \even_{r+1}$
this implies that $a_{i+1} \in \odd_{r+1}$.
Thus the edge $(p_{i+1}, p_i)$ in $G_M$ implies that there is an $\odd_{r+1} 
\un_{r+1}$ edge $(a_{i+1}, p_i)$ in the graph $G'$.
However, such an $\odd_{r+1} \un_{r+1}$ edge
should have been deleted by Irving et al.'s algorithm 
Hence such an edge cannot be present in $G_M$
contradicting the fact that $p \in \un_{r+1}$. Thus, every post $p$ belonging 
to a sink component belongs to  $\even_{r+1}$.
\end{enumerate}

The above proof  along with the fact that $\p \cap \odd_{r+1} = \emptyset$ immediately implies that a post $p$ belongs to a non-sink 
component iff $p \in \un_{r+1}$. 
\item First assume that $p$ has a path $T$ to a sink
and $w(T) = 0$. Our goal is to show that $p \in \even_1 \cap \ldots \cap \even_{r+1}$. Since $p$ has a path to a sink, $p \in \EE_{r+1}$.
Assume for the sake of contradiction that $p\in \odd_i\cup \U_i$ for some 
$i\leq r$.
Let $M' = M \cdot T$ be the matching obtained by switching along the path $T$.
Since $T$ has zero weight, from Lemma~\ref{lem:0wt-path},
we know 
that the matching $M'$ is a rank-maximal matching in $G$. Note that, since $q$
is unmatched in $M$, $p$ is unmatched in $M'$. Thus we have obtained
 a rank-maximal matching $M'$ in $G$ which leaves $p$ unmatched.
By the
invariants of Irving et al.'s algorithm, mentioned in Section~\ref{sec:prelim}, 
we know that every vertex belonging to 
$\odd_i \cup \U_i$ remains matched
in every rank-maximal matching of $G$. However, we have already obtained a 
matching, namely $M'$, which leaves $p$ unmatched.
This contradicts the assumption that $p\in \odd_i\cup\U_i$ for some $i$,
and completes the proof that $p \in \even_1 \cap \ldots \cap \even_{r+1}$.

Finally, consider  the other direction, that is assume that $p \in \even_1 \cap \ldots \cap \even_{r+1}$
and $p$ has a path $T$ to a sink. To show that $w(T) = 0$, we use arguments similar to proof of Lemma~\ref{lem:rmm-cycle}.
\end{enumerate}
This completes the proof of the Lemma.
\qed
\end{proof}
\begin{appendix-theorem}{\ref{thm:another-rmm}}
Every rank-maximal matching $M'$ in $G$ can be obtained from $M$ by applying to $M$ 
vertex-disjoint switching paths and switching cycles in $G_M$.
\end{appendix-theorem}
\begin{proof}
Consider any rank-maximal matching  $M'$ in $G$. We show that $M'$ can be obtained from $M$ by applying a set
of vertex-disjoint switching paths and switching cycles of $G_M$.
Consider $M \oplus M'$ which is a collection of vertex-disjoint paths and 
cycles in $G$. Also note that the cycles and paths contain alternating edges
of $M$ and $M'$.  
We show that the paths and
cycles in $M\oplus M'$ are switching paths and switching cycles in $G_M$.

From the invariants of Irving et al.'s algorithm mentioned in Section~\ref{sec:prelim}, all the
edges of $M$ and $M'$ are also present in $G'$. 
A cycle in $M\oplus M'$ has alternating edges of $M$
and $M'$, and hence has a corresponding directed cycle in $G_M$. As proved in 
Lemma~\ref{lem:rmm-cycle}, every cycle in $G_M$ is a switching cycle.

Now we consider paths in $M\oplus M'$. 
All the paths are of even length, since all the rank-maximal matchings are of 
the same cardinality.
Let $T_G =\langle p_1, a_1, \ldots,p_k, a_k, p_{k+1}\rangle$ be any even-length
path in $M \oplus M'$  with $p_{k+1}$ unmatched in $M$ and $p_1$ unmatched in 
$M'$.
For every $1 \le i \le k$, let
 $M(p_i) = a_i$.
It is easy to see that the path $T = \langle p=p_1, p_2, \ldots, 
p_{k+1}=p'\rangle$ is
present in $G_M$ and it ends in a sink $p'$.
Our goal is to show that $w(T) = 0$. For this, we prove that $p_1\in \EE_1
\cap\ldots \cap \EE_{r+1}$. 
Note that $M'$ is a rank-maximal matching in $G$ and
$M'$ leaves the post $p= p_1$ unmatched. As every post in $\odd_i\cup \un_i$
for any $i$ is matched in every rank-maximal matching, $p_1\notin \odd_i\cup 
\un_i$ for $1\leq i\leq r+1$. Therefore $p_1 \in \even_1 \cap \ldots \cap 
\even_{r+1}$; 
Thus, using Lemma~\ref{lem:switch-path}, we can conclude that the path $T$ has weight $w(T) = 0$ in $G_M$, and hence is a switching path in $G_M$.

Applying these switching paths
and cycles to $M$ gives us the desired matching $M'$, thus completing the proof.
\qed
\end{proof}

\section{Proofs from Section~\ref{sec:count}}\label{sec:app-count}
\label{app:sec-count}
\begin{appendix-lemma}{\ref{lem:pm-vs-rmm}}
Let $H$ be a 3-regular bipartite graph and let $G$ be the corresponding rank-maximal
matchings instance constructed by the reduction in Section~\ref{sec:count}. 
A matching $M$ is a perfect in $H$ iff $M$ is a rank-maximal matching in $G$.
\end{appendix-lemma}
\begin{proof}
A perfect matching $M$ in $H$ can be extended in a unique way to a perfect 
matching $M'$ in $G$ as follows: $M' = M \cup  \mbox { }\{(ad_i,pd_i)\mid 1\leq i\leq n-3\}$. 
The matching $M'$ has the following $n$-tuple as its signature: $\sigma(M') = (n-2,1,1,\ldots,1)$. Hence, a 
rank-maximal matching in $G$ should have a signature that is at least as good
as $\sigma(M')$. We argue that $\sigma(M')$ is the best possible signature in $G$.

Consider the posts in $\p$ that are rank $1$ posts for some applicant $a \in \A$.
There are exactly $n-2$ such posts: the $n-3$ dummy posts $pd_1, \ldots, pd_{n-3}$
and one post $p_y$ such that $order(p_y) = 1$.
Therefore, any rank-maximal matching in $G$ cannot match more than $n-2$ applicants 
to their rank-1 posts. 
Moreover, all these $n-3$  posts are {\em odd} or {\em unreachable}
in the graph on rank $1$ edges and hence are always matched to applicants that treat them as
their rank $1$ posts. 
At each of the ranks $2 \leq i \leq n$, there is exactly one
post $p_y\in \p$ that is ranked $i$. Thus, it is easy to see that $\sigma(M')$ is the best possible signature for any 
matching in $G$. Therefore $M'$ is a rank-maximal matching in $G$. Further, 
note that  $M'$ was obtained by extending a perfect matching $M$ in $H$. This also implies that, for every perfect matching in $H$, there is 
a unique rank-maximal matching in $G$.

Now consider a rank-maximal matching $M$ in $G$. We claim that such a matching
has to include the edges $\{(ad_i,pd_i)\mid 1\leq i \leq n-3\}$. If not, then
for some $i = 1 \ldots n-3$ applicant $ad_i$ remains unmatched and therefore $\sigma(M) \prec (n-2, 1, \ldots, 1)$.
Similarly, to achieve the signature $(n-2, 1, \ldots, 1)$, all applicants $a_x$
and therefore all posts $p_y$ should be matched amongst themselves.
Thus, the matching $M' = M \setminus \{(ad_i,pd_i)\mid 1\leq i \leq n-3\}$ is
a perfect matching in $H$. 
This shows that there is a one to one correspondence between the rank-maximal 
matchings in $G$ and perfect matching in $H$.
\qed
\end{proof} 

\begin{appendix-lemma}{\ref{lem:fpras-correct}}
Let $G$ be the rank-maximal matchings instance and let $H$  and $I$ be the rank-maximal matchings
instance and the  bipartite perfect matchings instance respectively as constructed  in Section~\ref{sec:fpras}. Then, the following hold:
\begin{enumerate}
\item\label{itm:fI} Corresponding to each rank-maximal matching $M$ in $G$, 
there are exactly $k!$ distinct rank-maximal matchings in $H$.
\item\label{itm:pm} 
Each rank-maximal matching in $H$ matches all the applicants and posts, and
all its edges appear in $I$. Hence it is a perfect matching in the instance $I$.\item\label{itm:nonrmm} A matching in $G$ that is not rank-maximal has no
corresponding perfect matching in $I$.
\end{enumerate}
\end{appendix-lemma}
\begin{proof}
\begin{enumerate}
\item It is easy to see that Irving et al.'s algorithm proceeds on $H$ 
exactly as on 
$G$ for $r+1$ iterations. Hence signature of a rank-maximal matching in $H$
is the same as that in $G$ for first $r+1$ co-ordinates. In $(r+2)$nd iteration,
the newly added applicants and their edges are considered. 
This iteration has a complete bipartite graph on rank $r+2$ edges, with $k$ 
applicants on left and $\mid \EE_1\cap\ldots\cap\EE_{r+1}\mid$ posts on right,
with exactly $k$ posts unmatched. So there are $k!$ ways of matching these 
newly added 
applicants amongst the unmatched posts. Thus, corresponding to each rank-maximal
matching constructed in the first $r+1$ iterations (and hence in $G$), there 
are $k!$ 
rank-maximal matchings in $G_{r+2}$. 
\item This is immediate from the construction.
\item
We show that every perfect matching in $I$ is a rank-maximal matching in $H$.
Let there be a perfect matching $M$ in $I$ which is not rank-maximal in $H$.
Consider a rank-maximal matching $N$ in $H$. As $M$ and $N$ are both perfect
matchings, $M\oplus N$ is a collection of vertex-disjoint cycles with alternate
edges of $M$ and $N$. Hence these cycles are switching cycles in the switching 
graph $H_N$. But
all the switching cycles in $H_N$ have weight $0$. By Lemma~\ref{lem:0wt-path},
this contradicts the assumption that $M$ is not rank-maximal.

\end{enumerate} 
This completes the proof of the lemma.
\qed
\end{proof}
\section{Details from Section~\ref{sec:popular}}\label{sec:app-popular}
\label{app:sec-popular}
\paragraph{The algorithm for popular rank-maximal matching}

\begin{enumerate}
\item Make $M$ $\A$-complete by adding dummy last-resort posts at rank $r+1$
as described in Section~\ref{sec:prelim}. Here $r$ is the maximum length of any
preference list in $G$. 
Remove those edges $(p_i,p_j)$ from $G_M$, where $(M(p_i),p_j)$
is not a rank-maximal pair. 
\item Construct a re-weighted graph $\tilde{G}_M$ by replacing every positive-weight in $G_M$
by $+1$ and negative weight by $-1$.
\item For each sink-component $\X_i$ in the $\tilde{G}_M$, add a source vertex $s_i$ 
and a sink $t_i$. For every non-sink vertex $p \in \EE_1 \cap \ldots \cap 
\EE_{r+1}$ add an edge
$(s_i, p)$ of weight zero. For every sink $p$, add an edge 
$(p, t_i)$ of weight zero. 
Thus $s_i$ and $t_i$ are new source
and sink for the component $\X_i$.
For each non-sink component of $\tilde{G}_M$, 
choose an arbitrary vertex $v$ as source. 

\item Run Bellman-Ford algorithm from each source to find if there is a 
negative-weight cycle in any component or a negative-weight $s_i$ to $t_i$ path
in any sink-component $\X_i$. 
\item If there exists a negative-weight path $T$ or cycle $C$ in the above graph, then conclude that $M$ is not a popular rank-maximal matching.
Find such a cycle $C$ or path $T$ 
and output $M'=M\cdot C$ or $M'=M\cdot T$ respectively, as a rank-maximal 
matching more popular than $M$. Otherwise $M$ is a popular rank-maximal 
matching.
\end{enumerate}
\begin{appendix-lemma}{\ref{lem:popular}}
A given rank-maximal matching $M$ is popular if and only if there is no negative
weight path to a sink or a negative-weight cycle in the re-weighted switching 
graph.
\end{appendix-lemma}
\begin{proof}
We prove that, if there is a negative-weight cycle or a negative-weight path
to a sink in $\tilde{G}_M$, then a matching $N$ obtained by switching along 
such a path or cycle is more popular than $M$.

Consider a negative-weight cycle, say $C=(p_1,\ldots,p_k,p_1)$ in $\tilde{G}_M$. Then
All the applicants, such that weight of $(M(p_i),p_{i+1})$ is negative 
(positive), get a post of strictly better (worse) rank in $M\cdot C$. As $C$
has more edges of negative-weight than those of positive-weight, more 
applicants prefer $M\cdot C$ over $M$ than those who prefer $M$ over
$M\cdot C$. Thus $M\cdot C$ is more popular than $M$ and hence $M$ is not 
popular. Similar argument holds for a path to sink.

Now, let there be a rank-maximal matching $N$ that is more popular than $M$.
The matching $N$ can be obtained from $M$ by application of a set $S$ of 
vertex-disjoint 
switching cycles/paths. Those applicants $a$ such that $M(a)$ is not a part
of any switching cycle or path are indifferent between $M$ and $N$. 
An applicant $a$ prefers $N$ ($M$) over $M$ ($N$) when he gets a better 
(worse) ranked
post in $N$ ($M$) than in $M$ ($N$). But $a$ gets $M(a)$ in $M$ and a post $p$ 
in $N$, where
$(M(a),p)$ is a part of a switching cycle/path in $S$. 
But then this edge should have
a negative (positive) weight in $G_M$ and hence a $-1$ ($+1$) weight in $\tilde{G}_M$.
As more applicants prefer $N$ over $M$, there should be more negative-weight
edges than positive-weight edges in the switching paths/cycles in $S$. 
Hence there is at least one switching cycle/path in $S$ that has more negative-
weight edges than positive-weight edges, which is a negative-weight cycle/path
in $\tilde{G}_M$.

Note that removal of those edges from $G_M$ which do not correspond to rank-maxi
mal pairs ensures that a negative-weight cycle, if present, is reachable from
one of the chosen sources.
\qed
\end{proof} 
\end{appendix}

\end{document}